%% file: Subarchitectures.tex
\documentclass[conference]{IEEEtran}

\input{preamble.tex}

\title{Practical Subarchitectures for \\ Optimal Quantum Layout Synthesis
\ifdefined\doubleblind
\else
\thanks{Part of this research is funded by the Innovation Fund Denmark.}
\fi
}

\ifdefined\doubleblind
\author{\IEEEauthorblockN{Anonymous Author(s) }}
\else
\author{\IEEEauthorblockN{1\textsuperscript{st} Kostiantyn V. Milkevych }
\IEEEauthorblockA{\textit{Department of Computer Science} \\
\textit{Aarhus University}\\
Aarhus, Denmark \\
\href{mailto:kmilkevych@gmail.com}{kmilkevych@gmail.com}}
\and
\IEEEauthorblockN{2\textsuperscript{nd} Jaco van de Pol}
\IEEEauthorblockA{\textit{Department of Computer Science} \\
\textit{Aarhus University}\\
Aarhus, Denmark \\
\url{https://orcid.org/0000-0003-4305-0625}}
\and
\IEEEauthorblockN{3\textsuperscript{rd} Irfansha Shaik}
\IEEEauthorblockA{\textit{Kvantify Aps.}\\
Copenhagen, Denmark}
\IEEEauthorblockA{\textit{Dept. of Computer Science,
Aarhus University}\\
Aarhus, Denmark \\
\url{https://orcid.org/0000-0002-7404-348X}}
}
\fi

\begin{document}

\maketitle

\begin{abstract}
Quantum Layout Synthesis (QLS) maps a logical quantum circuit to a physical quantum platform.
Optimal QLS minimizes circuit size and depth, which is essential to reduce the noise on current
quantum platforms. Optimal QLS is an NP-hard problem, so in practice, one maps a quantum
circuit to a subset of the complete quantum platform. However, to guarantee optimality,
one still has to consider exponentially many subarchitectures.

We introduce an effective method to enumerate relevant subarchitectures.
This reduces the number of considered subarchitectures, 
as well as the number of expensive subgraph isomorphism checks, 
thus boosting Optimal QLS with subarchitectures.
To do so, we assume a fixed number of ancilla qubits that can be used in the mapping. 
We guarantee optimality of the quantum layout, for the selected ancilla bound.

We evaluate our technique on a number of benchmarks and compare it with state-of-the-art 
Optimal QLS tools with and without using subarchitectures.
\end{abstract}

\begin{IEEEkeywords}
Quantum Layout Synthesis, Qubit Allocation, Subgraph isomorphism, Subgraph enumeration.
\end{IEEEkeywords}

\section{Introduction}
To execute quantum algorithms, they must first be compiled for a particular quantum computer platform.
One of the last steps in the compilation is \emph{layout synthesis}. 
Here the ``logical'' qubits of a quantum circuit are \emph{allocated} to ``physical'' qubits,
and \emph{routed} by inserting SWAP-operations. This is needed, since many (superconducting) 
platforms restrict binary operations to connected physical qubits, specified by a \emph{coupling graph}.

Since each inserted SWAP causes additional noise when executing the circuit,
the goal of \emph{optimal layout synthesis} (OLS) is to insert the minimal number of SWAP 
operations.\footnote{Another popular metric is to minimize the depth of the circuit.}
OLS is an NP-complete problem \cite{Siraichi_2018}, so many approaches have been suggested.
\emph{Heuristic} methods~\cite{Li_2019, Zulehner_2019}, like eager approaches or A* search, are fast but lead to highly
non-optimal solutions~\cite{Tan_2021}. \emph{Exact} methods~\cite{Wille_2019,Shaik_2023,Shaik_2024,WanHsuan_2023} 
are often based on leveraging domain-independent solvers for SAT, SMT, or planning. These methods provide
a provably optimal layout, but don't scale for larger instances. Several \emph{near-optimal}
\cite{Tan_2020,WanHsuan_2023} solutions have been proposed, to find high-quality solutions faster for many cases.

A prominent example is the use of \emph{subarchitectures} \cite{Wille_2019}, in which a quantum circuit
of $n$ qubits is mapped onto a subgraph of $n$ physical qubits from the full architecture.
The smaller mapping tasks are simpler, and only the best result is kept. 
This method is only near-optimal, since it doesn't consider subarchitectures with more than
$n$ qubits, so-called ancillaries, which may admit even less SWAP operations~\cite{Peham_2023}.
Those authors demonstrated that it is very difficult to determine the required size
of the subarchitecture. 
They proposed a near-optimal solution, reducing the number of considered subarchitectures
by combining them into larger \emph{coverings}. Their method involves enumerating all
subgraphs of the coupling graph, which is unfeasible for larger quantum platforms.
So far, subarchitectures have hardly been effective for OLS.

\paragraph*{Contribution}
In this paper, we propose an alternative approach and fast algorithms for subarchitecture
selection. Our approach guarantees to find an optimal layout among those using at most $k$
ancillaries, where $k$ is a parameter. We only consider \emph{maximal} subarchitectures,
in the sense of subgraph isomorphism. To compute these candidates quickly, we propose a
combination of graph hashing~\cite{Shervashidze_2011} and the direct enumeration of 
connected subgraphs of size $n+k$~\cite{Karakashian_2013}.

We implement our method 
\ifdefined\doubleblind \else
in the tool Q-Synth v4\footnote{\url{https://github.com/irfansha/q-synth}}
\fi
and compare it with using subarchitecture coverings in the tool MQT QMAP~\cite{Wille_2023,Qmap}, 
and with the SAT-based tool Q-Synth~\cite{Shaik_2024,Q-Synth} on quantum platforms
with 100+ qubits. Optimal solutions are found for larger instances and in less 
time than previous approaches.
Our experiments demonstrate that we made subarchitectures practical for 
\emph{optimal quantum layout synthesis}.

\section{Preliminaries}
\paragraph*{Graph terminology}
A \emph{graph} $G=(V,E)$ has a set of \emph{vertices} $V$ and a set of \emph{edges} $E\subseteq V\times V$.
We consider undirected graphs: $\forall v,w{\,\in\,} V(vEw\iff wEv)$. We write $G'\sqsubseteq G$ if
$G'=(V',E')$ is a \emph{subgraph} of $G$, i.e.\ $V'\subseteq V$ and $E'\subseteq E$. 
Subgraph $G'$ is \emph{induced} in $G$ if $\forall v,w{\,\in\,}V'(vE'w \iff vEw$). We write $G\simeq G'$ if $G$ and $G'$
are \emph{isomorphic}, i.e.\ there is a bijection $h:V\to V'$ s.t.\ 
$\forall v,w{\,\in\,}V (vEw \iff h(v)E'h(w))$. $G'$ is \emph{subgraph isomorphic} to $G$,
if $G'\simeq G''$ for some $G''\sqsubseteq G$. A path of length $\ell$ in $G=(V,E)$ is a sequence
of vertices $v_1;\ldots;v_\ell$ such that $v_i\in V$ and $v_i \,E\, v_{i+1}$ ($\forall i$). A graph is
\emph{connected} if there exists a path between any two vertices in $V$.

\paragraph*{Quantum Circuits}
A \emph{quantum circuit} operates on $n$ (logical) qubits $L=\{q_1,\ldots,q_n\}$.
A quantum circuit (of size $m$) is simply a sequence of quantum gates $\Gamma = g_1;\ldots;g_m$.
We assume that the circuit has already been compiled to the gate set of some quantum platform,
consisting only of unary gates and the binary gate $\CX$ (Controlled NOT). A unary gate
is of the form $U(q_i)$, where $U$ is the gate name. Each binary gate is of the form
$\CX(q_i,q_j)$, $i\neq j$, with control qubit $i$ and target qubit $j$. The semantics of these gates
do not play a role in this paper; we refer the interested reader to~\cite{Nielsen_2010}. 
Figure~\ref{fig:examplea} shows a quantum circuit on 4 qubits, consisting of 4 CX-gates.

\begin{figure*}[tbhp]
    \begin{subfigure}[b]{0.2\textwidth}\centering
    \scalebox{1}{
      \Qcircuit @C=0.7em @R=0.3em @!R { \\
      \lstick{{q}_{1} :} & \ctrl{1} & \qw      & \qw      & \targ    & \qw\\
      \lstick{{q}_{2} :} & \targ    & \ctrl{1} & \qw      & \qw      & \qw\\
      \lstick{{q}_{3} :} & \qw      & \targ    & \ctrl{1} & \qw      & \qw\\
      \lstick{{q}_{4} :} & \qw      & \qw      & \targ     & \ctrl{-3} & \qw\\
      &
      }}
      \medskip
      \caption{Logical quantum circuit:\\ with 4 qubits, 4 CX gates\label{fig:examplea}}
    \end{subfigure}
\quad
\begin{subfigure}[b]{0.16\textwidth}\centering
    \scalebox{0.7}{
        \begin{tikzpicture}
        \draw[line width=1.25pt] (18:1.25) arc[start angle=18, end angle=378, radius=1.2]
            foreach \r in {1,2,...,5} {
                node[pos={1/5*\r}, circle, draw, fill=white,inner sep=2pt] {$Q_\r$}
            };
        \end{tikzpicture}}
        \medskip
        \caption{Quantum platform:\\ with 5 physical qubits\label{fig:exampleb}}
    \end{subfigure}
\qquad
%
\begin{subfigure}[b]{0.26\textwidth}\centering
    \scalebox{0.95}{
        \Qcircuit @C=0.8em @R=0.3em @!R {
        \lstick{q_1{\to} Q_1 :} & \ctrl{1} & \qw      & \qswap     & \qw        & \qw       & \qw & q_2\\
        \lstick{q_2{\to} Q_2 :} & \targ    & \ctrl{1} & \qswap\qwx & \qswap     & \qw       & \qw & q_3\\
        \lstick{q_3{\to} Q_3 :} & \qw      & \targ    & \ctrl{1}   & \qswap\qwx & \targ     & \qw & q_1\\
        \lstick{q_4{\to} Q_4 :} & \qw      & \qw      & \targ      & \qw        & \ctrl{-1} & \qw & q_4\\
        \lstick{         Q_5 :} & \qw      & \qw      & \qw        & \qw        & \qw       & \qw
    }}
    \medskip
    \caption{Mapped circuit without ancillas:\\ using 2 SWAP gates\label{fig:examplec}}
\end{subfigure}
\quad
\begin{subfigure}[b]{0.26\textwidth}\centering
    \scalebox{0.95}{
        \Qcircuit @C=0.8em @R=0.3em @!R {
        \lstick{q_1{\to} Q_1 :} & \ctrl{1} & \qw      & \qw      & \qswap         & \qw      & \qw & \\
        \lstick{q_2{\to} Q_2 :} & \targ    & \ctrl{1} & \qw      & \qw            & \qw      & \qw & q_2\\
        \lstick{q_3{\to} Q_3 :} & \qw      & \targ    & \ctrl{1} & \qw            & \qw      & \qw & q_3\\
        \lstick{q_4{\to} Q_4 :} & \qw      & \qw      & \targ    & \qw            & \ctrl{1} & \qw & q_4\\
        \lstick{         Q_5 :} & \qw      & \qw      & \qw      & \qswap\qwx[-4] & \targ    & \qw & q_1
    }}
    \medskip
    \caption{Mapped with 1 ancilla qubit:\\ using 1 SWAP gate only\label{fig:exampled}}
\end{subfigure}
\caption{Example of Quantum Layout Synthesis with and without ancillas, adapted from~\cite{Peham_2023} \label{fig:mapped}}
\end{figure*}
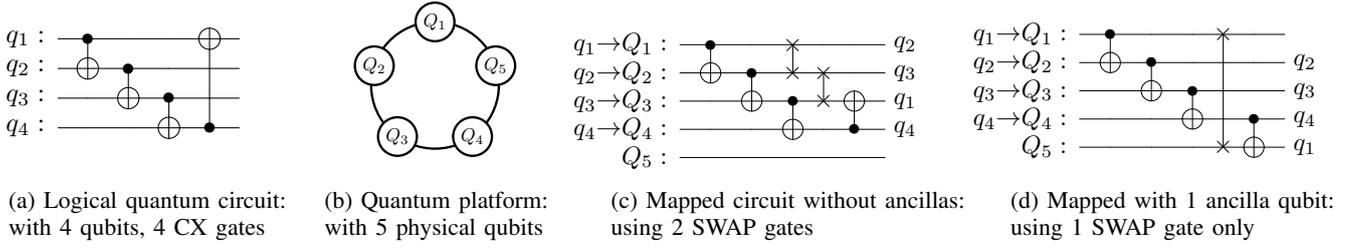

    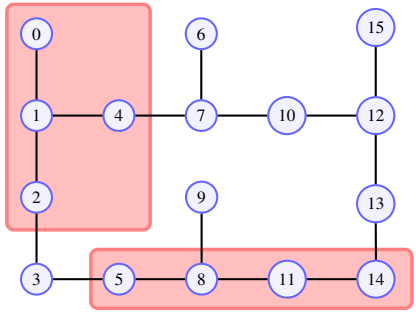
\begin{figure}[b]
    \centering
    \scalebox{0.65}{
    \begin{tikzpicture}[
    circlenode/.style={circle, draw=blue!60, fill=blue!5, very thick},
    ]
    \begin{pgfonlayer}{background}
        \draw[region, red!50, fill=red!25] (-0.6,0.6)--(2.3,0.6)--(2.3,-4)--(-0.6,-4)--cycle; 
      \end{pgfonlayer}
    
      \begin{pgfonlayer}{background}
        \draw[region, red!50, fill=red!25] (1.1,-4.4)--(7.6,-4.4)--(7.6,-5.6)--(1.1,-5.6)--cycle; 
      \end{pgfonlayer}
    \node[circlenode] (n0)                  {0};
    \node[circlenode] (n1)  [below=of n0]   {1};
    \node[circlenode] (n2)  [below=of n1]   {2};
    \node[circlenode] (n3)  [below=of n2]   {3};
    \node[circlenode] (n4)  [right=of n1]   {4};
    \node[circlenode] (n5)  [right=of n3]   {5};
    \node[circlenode] (n7)  [right=of n4]   {7};
    \node[circlenode] (n6)  [above=of n7]   {6};
    \node[circlenode] (n8)  [right=of n5]   {8};
    \node[circlenode] (n9)  [above=of n8]   {9};
    \node[circlenode] (n10) [right=of n7]   {10};
    \node[circlenode] (n11) [right=of n8]   {11};
    \node[circlenode] (n12) [right=of n10]  {12};
    \node[circlenode] (n13) [below=of n12]  {13};
    \node[circlenode] (n14) [right=of n11]  {14};
    \node[circlenode] (n15)  [above=of n12] {15};
    \draw[very thick, -] (n0.south) -- (n1.north);
    \draw[very thick, -] (n1.east) -- (n4.west);
    \draw[very thick, -] (n4.east) -- (n7.west);
    \draw[very thick, -] (n7.east) -- (n10.west);
    \draw[very thick, -] (n10.east) -- (n12.west);
    \draw[very thick, -] (n12.north) -- (n15.south);
    \draw[very thick, -] (n3.east) -- (n5.west);
    \draw[very thick, -] (n5.east) -- (n8.west);
    \draw[very thick, -] (n8.east) -- (n11.west);
    \draw[very thick, -] (n11.east) -- (n14.west);
    \draw[very thick, -] (n1.south) -- (n2.north);
    \draw[very thick, -] (n2.south) -- (n3.north);
    \draw[very thick, -] (n12.south) -- (n13.north);
    \draw[very thick, -] (n13.south) -- (n14.north);
    \draw[very thick, -] (n8.north) -- (n9.south);
    \draw[very thick, -] (n6.south) -- (n7.north);
    \end{tikzpicture}}
    \caption{IBM Guadalupe -- 16 qubits (average degree 2),\\
    with 2 maximal subarchitectures of 4 qubits
    \label{fig:guadalupe}}
    \end{figure}

    \begin{figure}[b]
    \centering
    \scalebox{0.65}{
    \begin{tikzpicture}[
    circlenode/.style={circle, draw=blue!60, fill=blue!5, very thick},
    ]
    \begin{pgfonlayer}{background}
        \draw[region, red!50, fill=red!25] (2.8,-1.1)--(5.6,-1.1)--(5.6,-4.0)--(2.8,-4.0)--cycle; 
      \end{pgfonlayer}
    \node[circlenode] (n0) {0};
    \node[circlenode] (n1) [right=of n0] {1};
    \draw[very thick, -] (n0.east) -- (n1.west);
    \node[circlenode] (n2) [right=of n1] {2};
    \draw[very thick, -] (n1.east) -- (n2.west);
    \node[circlenode] (n3) [right=of n2] {3};
    \draw[very thick, -] (n2.east) -- (n3.west);
    \node[circlenode] (n4) [right=of n3] {4};
    \draw[very thick, -] (n3.east) -- (n4.west);
    \node[circlenode] (n5) [below=of n0] {5};
    \draw[very thick, -] (n0.south) -- (n5.north);
    \node[circlenode] (n6) [below=of n1] {6};
    \draw[very thick, -] (n1.south) -- (n6.north);
    \draw[very thick, -] (n5.east) -- (n6.west);
    \node[circlenode] (n7) [below=of n2] {7};
    \draw[very thick, -] (n2.south) -- (n7.north);
    \draw[very thick, -] (n6.east) -- (n7.west);
    \node[circlenode] (n8) [below=of n3] {8};
    \draw[very thick, -] (n3.south) -- (n8.north);
    \draw[very thick, -] (n7.east) -- (n8.west);
    \node[circlenode] (n9) [below=of n4] {9};
    \draw[very thick, -] (n4.south) -- (n9.north);
    \draw[very thick, -] (n8.east) -- (n9.west);
    \node[circlenode] (n10) [below=of n5] {10};
    \draw[very thick, -] (n5.south) -- (n10.north);
    \node[circlenode] (n11) [below=of n6] {11};
    \draw[very thick, -] (n6.south) -- (n11.north);
    \draw[very thick, -] (n10.east) -- (n11.west);
    \node[circlenode] (n12) [below=of n7] {12};
    \draw[very thick, -] (n7.south) -- (n12.north);
    \draw[very thick, -] (n11.east) -- (n12.west);
    \node[circlenode] (n13) [below=of n8] {13};
    \draw[very thick, -] (n8.south) -- (n13.north);
    \draw[very thick, -] (n12.east) -- (n13.west);
    \node[circlenode] (n14) [below=of n9] {14};
    \draw[very thick, -] (n13.east) -- (n14.west);
    \node[circlenode] (n15) [below=of n10] {15};
    \draw[very thick, -] (n10.south) -- (n15.north);
    \node[circlenode] (n16) [below=of n11] {16};
    \draw[very thick, -] (n11.south) -- (n16.north);
    \draw[very thick, -] (n15.east) -- (n16.west);
    \node[circlenode] (n17) [below=of n12] {17};
    \draw[very thick, -] (n16.east) -- (n17.west);
    \node[circlenode] (n18) [below=of n13] {18};
    \draw[very thick, -] (n13.south) -- (n18.north);
    \node[circlenode] (n19) [below=of n14] {19};
    \draw[very thick, -] (n14.south) -- (n19.north);
    
    \draw[very thick, -] (n3) -- (n9);
    \draw[very thick, -] (n4) -- (n8);
    \draw[very thick, -] (n5) -- (n11);
    \draw[very thick, -] (n6) -- (n10);
    \draw[very thick, -] (n11) -- (n17);
    \draw[very thick, -] (n12) -- (n16);
    \draw[very thick, -] (n13) -- (n19);
    \draw[very thick, -] (n14) -- (n18);
    \draw[very thick, -] (n7) -- (n13);
    \draw[very thick, -] (n8) -- (n12);
    \end{tikzpicture}}
    \caption{IBM Tokyo -- 20 qubits (average degree 3.7),\\
    with 1 maximal subarchitecture of 4 qubits.
    \label{fig:tokyo}}
    \end{figure}
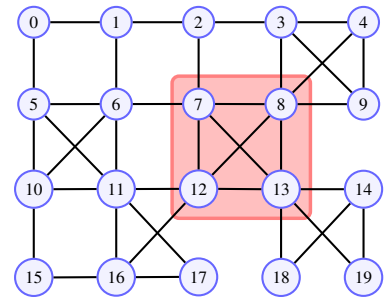

\paragraph*{Quantum Platform}
To execute a quantum circuit on a platform, the $n$ logical qubits are allocated to 
$p$ physical qubits $P=\{Q_1,\ldots,Q_p\}$. An allocation is an injective mapping 
$\alpha: L\to P$ from logical to physical qubits. Many (superconducting) quantum platforms
support $\CX$-gates on a subset of physical qubits only. This is specified in the 
\emph{coupling graph} $(P,C)$. Fig.~\ref{fig:exampleb} shows a coupling graph on 5 qubits, connected in a cycle.
Executing gate $\CX(Q_i,Q_j)$ is only \emph{feasible} if the physical qubits are connected, i.e., $Q_i \,C\, Q_j$.

\paragraph*{Quantum Layout Synthesis}
Often, there is no single $\alpha$ that makes all $\CX$-gates feasible. 
In this case, the qubits must be rerouted by inserting binary SWAP gates in the circuit.
The effect of $\SWAP(Q_i,Q_j)$ on physical qubits is defined by:
\[\pi_{ij}(Q_i)=Q_j,\quad \pi_{ij}(Q_j)=Q_i,\quad \pi_{ij}(Q_k)=Q_k \quad (\forall k\neq i,j)\]

We can now define the \emph{Quantum Layout Problem}: 
it takes as \textbf{input} a quantum circuit $\Gamma$ on logical qubits $L$ and a coupling graph $(P,C)$ on physical qubits $P$.
The \textbf{output} is a quantum circuit $\Gamma'$ on $P$, possibly with SWAP-gates, and an initial qubit allocation $\alpha:L\to P$.
The two correctness criteria are:

\textbf{Equivalence:} If the mapped circuit is ``unmapped'', we obtain the input circuit,%
\footnote{This equivalence can be relaxed, by allowing permutations of independent or even commuting gates.
The results in this paper also apply under relaxations.}
i.e.\ $um(\Gamma', \alpha) = \Gamma$ (defined below).

\textbf{Feasibility:} All $\CX(Q_1,Q_2)$- and $\SWAP(Q_1,Q_2)$-gates in $\Gamma'$ are on connected qubits, i.e.\ $Q_1\,C\,Q_2$.

When unmapping the physical circuit, we must reverse the initial allocation $\alpha$. Since $\alpha:L\to P$ is an injective function, its inverse $\overline{\alpha}:P\to (L\cup\{\bot\})$ is uniquely defined. 
Also, the effect of SWAP-gates must be ``undone''.
For instance, in Fig.~\ref{fig:exampled}, initially we have $\overline{\alpha}(Q_1)=q_1$ and $\overline{\alpha}(Q_5)=\bot$. After the displayed $\SWAP(Q_1,Q_5)$, we obtain the qubit allocation $\alpha' = \pi_{1,5}\circ\alpha$. 
Its inverse gives $\overline{\alpha'}(Q_1)=\bot$ and
$\overline{\alpha'}(Q_5)=q_5$.

We formally define $um(\Gamma,\alpha)$, the ``unmap''-function by:
\begin{align*} 
    um(U(Q);\Gamma,\alpha) & = U(\overline{\alpha}(Q)); um(\Gamma,\alpha) \\
    um(\CX(Q_i,Q_j);\Gamma, \alpha) & = \CX(\overline{\alpha}(Q_i),\overline{\alpha}(Q_j));um(\Gamma,\alpha) \\
    um(\SWAP(Q_i,Q_j);\Gamma, \alpha) & = um(\Gamma,\pi_{ij}\circ\alpha)
\end{align*}

\paragraph*{Optimal Layout Synthesis}
Inserting SWAP-gates is necessary to execute the circuit on the platform, but it is quite costly: 
Usually, a SWAP gate is implemented by 3 $\CX$-gates:
$\SWAP(q_i,q_j) = \CX(q_i,q_j);\CX(q_j,q_i);\CX(q_i,q_j)$.
This considerably increases the time and noise, so in \emph{Optimal} Quantum Layout Synthesis,
the \textbf{objective} is to insert the minimal number of SWAP-gates to obtain an equivalent and feasible solution.%
\footnote{Another popular objective is to minimize the depth of the quantum circuit.
The implementation in {Q-Synth v4} is fully compatible with SAT-based depth-optimal 
Quantum Layout Synthesis from~\cite{SAT25}.
}

\section{Approach based on Subarchitectures}
As mentioned in the introduction, optimal quantum layout synthesis is NP-complete, which makes it unfeasible for large inputs. 
One way to remedy this is to try mapping a circuit of $n$ nodes onto a subarchitecture, which is simply defined as an induced subgraph
$(P',C')\sqsubseteq (P,C)$ of the full platform. This introduces a new problem, since there are too many subgraphs to consider. 
We propose to restrict attention to subgraphs of $n+k$ nodes, where $k$ is a fixed number of ancilla qubits. 
To further restrict the candidates, we only consider \emph{connected} subgraphs of size $n+k$ that
are \emph{maximal} in the subgraph isomorphism ordering.
See Fig.~\ref{fig:guadalupe},~\ref{fig:tokyo} for two example quantum platforms with 16 resp.\ 20 qubits. 
These graphs have many 4-subgraphs (cf.\ Table~\ref{tab:subarchitectures}), but there are only 2 resp.\ 1 
connected 4-subgraphs that are maximal w.r.t.\ subgraph isomorphism.

The following theorem ensures that we do not miss the optimal solution \emph{among the subarchitectures with $n+k$ qubits}. 

\begin{theorem}
Assume quantum circuit $\Gamma$ on logical qubits $L$ can be mapped onto a platform with coupling graph $(P',C')$ using $\ell$ SWAPs.
If $(P',C')$ is subgraph isomorphic with $(P,C)$, then $\Gamma$ can also be mapped onto $(P,C)$ using $\ell$ swaps.
\end{theorem}
\begin{proof}
Let $\alpha:L\to P'$ be the initial allocation for an equivalent and feasible mapping of $\Gamma$ to a physical quantum circuit $\Gamma'$ with $\ell$ SWAPs onto $(P',C')$.
From the subgraph isomorphism, we obtain 
$(P',C')\simeq (P'',C'')\sqsubseteq (P,C)$. 
Assume $h:P'\to P''$ is an isomorphism between $(P',C')$ and $(P'',C'')$.
We claim that $h\circ\alpha:L\to P$ is an initial allocation, which yields an equivalent and feasible
circuit $h(\Gamma')$ on $(P,C)$.
\begin{itemize}
\item $\alpha:L\to P'$ is injective and so is $h$, so also $h\circ\alpha:L\to P''$ is injective. $P''\subseteq P$, so $h\circ\alpha:L\to P$ is injective too.
\item By induction on $\Gamma'$, we obtain $um(h(\Gamma'),h\circ\alpha) = um(\Gamma',\alpha) = \Gamma$.
Hence, $h(\Gamma')$ is equivalent to $\Gamma$.
\item All binary gates on $Q_1, Q_2$ in $\Gamma'$ satisfy $Q_1\,C\,Q_2$, so $h(Q_1)\,C''\,h(Q_2)$ by isomorphism, so $h(Q_1)\, C\, h(Q_2)$ since $C''\subseteq C$.
Hence, $h(\Gamma')$ is feasible on $(P,C)$.
\end{itemize}
\end{proof}

The reverse is not true, as shown in Fig.~\ref{fig:mapped}:
Mapping the logical circuit on the cyclic platform of Fig.~\ref{fig:exampleb}, 
needs 2 SWAPs when only 4 physical qubits are used (Fig.~\ref{fig:examplec}).
With ancillas (Fig.~\ref{fig:exampled}), 1 SWAP is sufficient,
now using 5 qubits.
We stress that computing an optimal solution on $n+k$ qubits might be suboptimal when considering more than $k$ ancillas.
To find out, one has to increase $k$ and repeat the procedure.

\subsection{Algorithm for finding maximal subarchitectures}
We have reduced our task to enumerating all \emph{maximal} subgraphs of a graph $G=(V,E)$ of a fixed size $k$.
We wish to minimize the number of expensive subgraph isomorphism checks. 
From now on, we assume (without loss of generality%
\footnote{If the coupling map is not connected, one can always decompose the OLS problem to smaller problems on maximal connected subarchitectures.}) 
that the coupling graph is connected.
The coupling map of existing quantum platforms is always connected and moreover, it is usually (close to) planar and has a low degree.

\noindent We then proceed in three steps, see Alg.~\ref{alg:subarchitectures}:
\begin{enumerate}
    \item We enumerate all \emph{connected} subgraphs of size $k$,
    \item We remove isomorphic duplicates, using graph-hashing,
    \item We test the remaining pairs for subgraph isomorphism.
\end{enumerate}

\begin{algorithm}[t]
\caption{Generate all maximal, connected $k$-subgraphs
\label{alg:subarchitectures}}
\begin{algorithmic}[1]
\Require Coupling graph $(P,C)$, size of subarchitectures $k$
\Ensure Returns the set of all maximal, connected $k$-subgraphs
\Procedure{MaxSubArch}{$P$,$C$,$k$}
\State $\mathit{MaxGraphs} \gets \emptyset$
\State $\mathit{Hashes} \gets \lambda h.\emptyset$ \Comment{non-isomorphic graphs with hash $h$}
\ForEach{$G \in \textsc{ConnectedSubGraphs}((P,C),k)$}
\State $h \gets \textsc{GraphHash}(G)$ \label{line:hash1}
\If{$\forall G'\in \mathit{Hashes}[h].\ G\not\simeq G'$}
    \State $\mathit{Hashes}[h] \gets \mathit{Hashes}[h]\cup \{G\}$ \label{line:hash2}
    \State $\mathit{candidate}\gets True$
    \ForEach{$G' \in \mathit{MaxGraphs}$} \label{line:max1}
    \If{$\textsc{subgraphIso}(G,G')$}
        \State{$\mathit{candidate}\gets False$}
        \State \textbf{break} \Comment{$G$ is not maximal}
    \ElsIf{$\textsc{subgraphIso}(G',G)$}
        \State $\mathit{MaxGraphs} \gets \mathit{MaxGraphs}\setminus \{G'\}$ \label{line:max2}
    \EndIf
    \EndForEach
    \If{$\mathit{candidate}$}
        \State $\mathit{MaxGraphs} \gets \mathit{MaxGraphs}\cup\{G\}$
    \EndIf
\EndIf
\EndForEach
\State \Return $\mathit{MaxGraphs}$
\EndProcedure
\end{algorithmic}
\end{algorithm}

\subsubsection{Enumerating all connected $k$-subgraphs}
\label{sec:connected}
A na\"ive algorithm to enumerate all connected induced subgraphs of size $k$ in $G=(V,E)$ would be 
to enumerate all $k$-subsets of $V$ and test the induced graphs for connectedness.
A sophisticated algorithm proposed in~\cite{Karakashian_2013} has the same worst-case time complexity, 
but it is more efficient for graphs with low degree, like most (superconducting) quantum platforms. 

Assume that $T$ is a {\em tree} with root $A$ whose children are $A_1,\ldots,A_\ell$. 
The set of connected $k$-subgraphs rooted at $A$ can be obtained as combinations of $\{A\}$ 
with the sets of $k_i$-connected subgraphs rooted at $A_i$, for any $k_i+\cdots+k_\ell = k-1$. 
The algorithm in~\cite{Karakashian_2013} generalizes this idea to the much more complicated 
case where $T$ is a spanning tree of an arbitrary graph $G$. 
We have implemented their algorithm, and evaluated it in our experiments.

\subsubsection{Graph-hashing as isomorphism filter}
\label{sec:isomorphic}
We use the Weisfeiler-Lehman graph-hashing technique~\cite{Shervashidze_2011}. 
This guarantees that isomorphic graphs get the same hash. Given non-isomorphic graphs, 
there is a small probability of hash collision. However, this probability can be tuned by a parameter.

In our implementation, we adapted the graph-hashing code from networkx based on~\cite{Shervashidze_2011}.
Instead of checking isomorphism between any two candidates, we compute the graph-hash of each candidate
(Alg.~\ref{alg:subarchitectures}, ll.~\ref{line:hash1}-\ref{line:hash2}), and check isomorphism only between candidates with the same hash.%
\footnote{Turning off this isomorphism check introduces a small chance that we miss some maximal subarchitectures. We observed this only once in our experiments.}
This reduces the number of isomorphism checks considerably.

\subsubsection{Subgraph Isomorphism Checks}
\label{sec:maximal}
At this point we have a list of non-isomorphic connected subgraphs of size $k$. We maintain a set of the maximal 
subgraph isomorphic ones seen so far, and compare every new subgraph with the previous maximal ones, keeping the new maximal ones only (Alg.~\ref{alg:subarchitectures}, ll.~\ref{line:max1}-\ref{line:max2}). This results
in $O(c^2)$ subgraph isomorphism checks, where $c$ is the number of non-isomorphic candidates. In our
implementation, we use the VF2 subgraph isomorphism check in rustworkx~\cite{rustworkx}, based on \cite{Cordella_2004}.

\subsection{Algorithm for layout synthesis with subarchitectures}

\begin{algorithm}[t]
    \caption{Mapping Strategy with Subarchitectures
    \label{alg:mapping}}
    \begin{algorithmic}[1]
        \Require Quantum Platform with coupling graph $(P,C)$
        \Require Logical circuit $\Gamma$ with $n\leq|P|$ qubits
        \Require A correct and optimal layout mapping function $\textsc{Map}$
    \Ensure Returns a correct and optimal layout mapping
    \Procedure{MapWithSubarch}{$(P,C)$, $\Gamma$}
    \State $\mathit{bound} \gets \infty$ \label{line:inf}
    \For{$k = n$ \textbf{to} $|P|$}
    \ForEach{$(P',C') \in \textsc{MaxSubArch}(P,C,k)$}
        \State \textbf{try} $(\Gamma',\alpha,S) \gets \textsc{Map}(\Gamma,(P',C'),\mathit{bound})$ \label{line:map}
        \If{Success}
            \If{$S=0$} \Return $(\Gamma',\alpha,S)$
            \Else \mbox{ }$\mathit{bound} \gets S-1$ \label{line:bound}
            \EndIf
        \EndIf
    \EndForEach
    \EndFor
    \State \Return the last successful $(\Gamma',\alpha,S)$
    \EndProcedure
    \end{algorithmic}
    \end{algorithm}

We now describe how any layout mapping procedure can be enhanced by using maximal subarchitectures (Alg.~\ref{alg:mapping}).
We assume that we have a blackbox function $\textsc{Map}(\Gamma,(P,C),b)$, which tries to find an optimal mapping of $\Gamma$ on $(P,C)$ that requires at most $b$ extra SWAP gates. If
successful, the procedure returns the triple $(\Gamma',\alpha,S)$, where $\Gamma'$ is
a correctly mapped circuit, $\alpha$ is the initial mapping, and $S$ indicates the minimum
number of SWAP gates required. Alg.~\ref{alg:mapping} provides a drop-in replacement
of $\textsc{Map}$, which is still  correct and optimal, but potentially more efficient
by using maximal subarchitectures.

Alg.~\ref{alg:mapping}, l.~\ref{line:map} tries to map $\Gamma$ on maximal subarchitectures $(P',C')$ with an increasing
number of ancillary qubits $k-n$. When it finds a solution with $S$ SWAP gates, it will
only try to find better ones, by setting the bound to $S-1$ (l.~\ref{line:bound}). This reduces the search time required for future calls to $\textsc{Map}$. 
In practice, one could quickly obtain an initial upperbound (instead of l.~\ref{line:inf}) by running a heuristic 
tool; we did not do so in our experiments.

The strategy of Alg.~\ref{alg:mapping} does not only guarantee
correctness and optimality, but it also returns a circuit with the minimal number
of ancillaries that allow a correct and optimal solution.
In practice, one would probably only try 0, 1 or 2 ancillaries, and break the outer
loop earlier. This would not guarantee the global optimal solution, but at least,
it guarantees the best solution with 0, 1 or 2 ancillaries.
    
\section{Experiments}

We carry out three experiments, evaluating our algorithm, and evaluating the improvement
it brings to the tools MQT QMAP (using subarchitecture coverings) and 
Q-Synth (using SAT to the full architecture).
All experiments are carried out on a single core of a 
48-core compute server with Intel Xeon Gold 6248R processors (3.0 GHz) and 384 GB internal memory.

\begin{table*}[tbh]
    \caption{Subarchitectures of size $k$ from quantum platforms with $|P|$ qubits.
    We show the number of all $\binom{|P|}{k}$ subarchitectures, and compute the connected, non-isomorphic, and 
    maximal ones w.r.t. $\sqsubseteq$.
    We also report the running times to compute the connected, non-isomorphic, and 
    maximal ones, and the total running time (in seconds).
    \label{tab:subarchitectures}}
    \begin{center}
    \begin{tabular}{|lrr||r|rrr||rrr|r|}
        \hline
        Platform & $|P|$ & $k$ & All Subsets & Connected & NonIso & Max & Connected (s) & NonIso (s) & Max (s) & Total (s)\\
        \hline\hline
        guadalupe & 16 & 4 & 1820 & 24 & 2 & 2
        & 0.0057 & 0.0027 & 0.0000 & 0.0084\\
        guadalupe & 16 & 8 & 12870 & 55 & 5 & 5
        & 0.0062 & 0.0095 & 0.0002 & 0.0159\\
        guadalupe & 16 & 12 & 1820 & 109 & 16 & 15
        & 0.0180 & 0.0283 & 0.0020 & 0.0483\\
        guadalupe & 16 & 16 & 1 & 1 & 1 & 1
        & 0.0204 & 0.0004 & 0.0000 & 0.0209\\
        \hline
        rigetti-16 & 16 & 4 & 1820 & 35 & 3 & 2
        & 0.0035 & 0.0032 & 0.0000 & 0.0068\\
        rigetti-16 & 16 & 8 & 12870 & 135 & 14 & 9
        & 0.0123 & 0.0239 & 0.0008 & 0.0370\\
        rigetti-16 & 16 & 12 & 1820 & 149 & 30 & 16
        & 0.0429 & 0.0380 & 0.0052 & 0.0862\\
        rigetti-16 & 16 & 16 & 1 & 1 & 1 & 1
        & 0.0522 & 0.0004 & 0.0000 & 0.0526\\
        \hline
        tokyo & 20 & 4 & 4845 & 179 & 6 & 1
        & 0.0087 & 0.0163 & 0.0001 & 0.0251\\
        tokyo & 20 & 8 & 125970 & 3883 & 207 & 18
        & 0.4077 & 0.7192 & 0.0366 & 1.1635\\
        tokyo & 20 & 12 & 125970 & 12402 & 2667 & 131
        & 11.3097 & 3.4248 & 5.1231 & 19.8577\\
        tokyo & 20 & 16 & 4845 & 1951 & 990 & 91
        & 269.0086 & 0.7124 & 1.6968 & 271.4178\\
        \hline
        sycamore & 53 & 4 & 292825 & 613 & 3 & 2
        & 0.0199 & 0.0535 & 0.0001 & 0.0734\\
        sycamore & 53 & 8 & $8.86{\times}10^{8}$ & 44226 & 51 & 9
        & 2.1086 & 7.8243 & 0.0047 & 9.9377\\
        sycamore & 53 & 12 & $2.67{\times}10^{11}$ & 3350459 & 2042 & 111
        & 398.6534 & 879.3076 & 7.2249 & 1285.1859\\
        sycamore & 53 & 16 & $1.48{\times}10^{13}$ & TO & TO & TO
        & TO & TO & TO & TO\\
        \hline
        rigetti-80 & 80 & 4 & 1581580 & 343 & 3 & 2
        & 0.0175 & 0.0309 & 0.0000 & 0.0485\\
        rigetti-80 & 80 & 8 & $2.90{\times}10^{10}$ & 6479 & 24 & 12
        & 0.2784 & 1.1118 & 0.0019 & 1.3921\\
        rigetti-80 & 80 & 12 & $6.02{\times}10^{13}$ & 140094 & 397 & 99
        & 7.2202 & 36.0356 & 0.4309 & 43.6867\\
        rigetti-80 & 80 & 16 & $2.70{\times}10^{16}$ & 2941151 & 8417 & 1034
        & 336.6784 & 1012.9459 & 235.7125 & 1585.3368\\
        \hline
        eagle & 127 & 4 & 10334625 & 278 & 2 & 2
        & 0.1424 & 0.0252 & 0.0000 & 0.1676\\
        eagle & 127 & 8 & $1.34{\times}10^{12}$ & 2100 & 8 & 4
        & 0.0957 & 0.3604 & 0.0008 & 0.4569\\
        eagle & 127 & 12 & $2.15{\times}10^{16}$ & 23446 & 54 & 19
        & 0.9733 & 5.9831 & 0.0239 & 6.9802\\
        eagle & 127 & 16 & $8.17{\times}10^{19}$ & 290551 & 412 & 140
        & 12.3763 & 98.4724 & 1.7108 & 112.5595\\
        \hline
    \end{tabular}
    \end{center}
    \end{table*}

\subsection{Effectiveness of computing maximal subarchitectures}
\label{sec:experiment1}

In the first experiment, we focus on Alg.~\ref{alg:subarchitectures} for computing maximal subarchitectures.
In particular, we evaluate the reduction in the number of subarchitectures obtained by the pipeline in our algorithm.
We took the coupling graphs of six well-known quantum platforms:
IBM Guadalupe (16 qubits), IBM Tokyo (20 qubits), IBM Eagle (127 qubits), Rigetti (16 and 80 qubits), 
and Google's Sycamore (53 qubits).
We first calculate the number of subsets of $k$ qubits from a platform with $|P|$ qubits simply
as $\binom{|P|}{k}$, where $k$ ranges over $\{4,8,12,16\}$.

Then we compute all connected subgraphs (Sec.~\ref{sec:connected}), using our own implementation in Python.
Subsequently, we filter the non-isomorphic ones using graph hashing (Sec.~\ref{sec:isomorphic}).
Finally, we use subgraph isomorphism checks to reduce this to only the maximal candidates (Sec.~\ref{sec:maximal}).

From the results in Tab.~\ref{tab:subarchitectures}, it is clear that enumerating all subsets 
(as is done in MQT QMAP with coverings) becomes quickly unfeasible. On the other hand, we
succeeded to generate the connected subgraphs for nearly all cases. The results also show the considerable reduction achieved by
eliminating isomorphic graphs using graph-hashing (column NonIso). Finally, the number of maximal candidates considered
for the mapping is quite manageable in most cases. We conclude that all steps of Alg.~\ref{alg:subarchitectures}
contribute to filtering relevant subarchitectures.

The rightmost columns of Tab.~\ref{tab:subarchitectures} show the running time of the various phases.
Most of the calculations finish within a fraction of a second, while two of them take up to half an hour.
Only calculating the 16-qubit connected subgraphs on Sycamore (53) timed out after 3 hours.
Note that the numbers and running time are not only determined by the size of the graph,
but also by the topology: IBM Guadalupe (Fig.~\ref{fig:guadalupe}) is sparser than IBM Tokyo (Fig.~\ref{fig:tokyo}), and indeed, for Tokyo the number of candidates and the running time increase faster.

\subsection{Comparison with subarchitecture coverings in MQT QMAP}

\input{figures/qmap-cover4.tex}

In the second experiment, we compare the effect of using \emph{maximal subarchitectures} as computed by Alg.~\ref{alg:subarchitectures} in MQT QMAP~\cite{Qmap}.
MQT QMAP uses SMT solving to find an optimal mapping of a quantum circuit on a quantum architecture~\cite{Wille_2019}. It comes with an implementation of \emph{subarchitecture coverings} as described in~\cite{Peham_2023}, which we refer to as \emph{QMAP-cov}.
Here the user supplies a number of desired subarchitectures that together ``cover'' a larger set of connected subarchitectures. QMAP-cov will (i) generate all subgraphs of the coupling graph; (ii) filter the non-isomorphic connected subgraphs; (iii) combine them into a set of larger subarchitectures with the desired cardinality using subgraph-isomorphism checks; and (iv) map the quantum circuit to each of those ``covering subarchitectures''.

In our modification of QMAP, which we call \emph{QMAP-max}, we replace step (i)-(iii) above by computing maximal subarchitectures of size $k$, resp.\ $k+1$ (corresponding to 0 and 1 ancillaries), following the steps in Alg.~\ref{alg:subarchitectures}.

We evaluate both methods on the same quantum circuits as used in~\cite{Peham_2023}, obtained by selecting all 5-qubit benchmarks from MQT BENCH~\cite{MQTbench23}.%
\footnote{We selected all scalable benchmarks with 5 qubits for IBM's native gates in Qiskit format optimization level 0.}
We run QMAP-cov with coverings of size 1, 2, 4 and 8, and we run QMAP-max with maximal subarchitectures using 0 or 1 ancilla qubits. We compare both procedures on running time and on the quality of the mapped circuits, i.e.\ the smallest number of added gates. The size of the coverings doesn't change the overall picture, so we only present the results for coverings of size 4.

We measure the running time in seconds of the original QMAP-cov (using coverings), and our modified QMAP-max (maximal subarchitectures) on all benchmarks, Fig.~\ref{fig:mqt_4:time}. 
The running time includes computing coverings / subarchitectures and the time needed for mapping by QMAP.
We also indicate the number gates inserted during layout mapping, Fig.~\ref{fig:mqt_4:gates}. 

The results indicate that using maximal subarchitectures scales better than using coverings (manifesting as fewer timeouts for QMAP-max). 
The results also show a clear trade-off: using maximal subarchitectures without ancillas (blue dots) is usually faster than using coverings. However, the resulting mapped circuits introduce a few more gates 
(blue dots in Fig.~\ref{fig:mqt_4:gates}).
On the other hand, using maximal subarchitectures with 1 ancilla (orange dots) is often slower than using coverings (although sometimes a lot faster, since computing coverings on the relatively dense Tokyo platform is expensive).
But the circuits mapped with 1 ancilla have fewer extra gates (orange dots in Fig.~\ref{fig:mqt_4:gates}). 
This trade-off is confirmed by the cumulative summary in Table~\ref{tab:cumcover} over all circuits that didn't time out; the underlying data is available in~\cite{KostyantinBSc}.
It shows the minimum, average and maximum runtime improvement and gate improvement over all solved benchmarks when using maximal subarchitectures instead of coverings.

To illustrate the trade-off on concrete larger examples, when mapping the quantum circuit PortfolioVQE (a Variational Quantum Eigensolver used in quantum chemistry) on Rigetti-16, \emph{QMAP-cov} introduces 120 extra gates using 804s, while \emph{QMAP-max with 1 ancillary} introduces only 98 extra gates but uses 7660s. To map on the IBM Tokyo platform, QMAP-max needs only 39 additional gates (using 7994s), while QMAP-cov times out. 
On the other hand, using \emph{QMAP-max with 0 ancillaries} adds 49 gates (using 19s) to map QFT-entangled (Quantum Fourier Transform) on IBM Tokyo, while \emph{QMAP-cov} adds only 22 gates (but using 7885s). 

Using coverings does not provide any optimality guarantee. For instance,
Table~\ref{tab:cumcover} shows that for some instance QMAP-max with 0 ancillaries
finds a solution with 6 fewer gates than QMAP-cov.
On the other hand, using maximal subarchitectures guarantees optimality for the specified number of ancillas.

\begin{table}[tb]
\caption{Cumulative numbers: improvement of QMAP-max with 0 and 1 ancillas over QMAP-cov with covering of size 4.\label{tab:cumcover}}
\centering
\begin{tabular}{|l||rrr||rrr|}
    \hline
    & \multicolumn{3}{|c||}{Running time improvement} & \multicolumn{3}{|c|}{Gate count impr.}\\
Ancillas & min (s) & avg (s) & max (s) & min & avg & max \\
    \hline   
0 anc (N=36) & -100.98 & 1013.86 & 9957.79 & -18 & -1.56 & 6 \\
1 anc (N=34) & -6855.24 & -179.22 & 9605.80 & 0 & 3.62 & 22 \\  
\hline
\end{tabular}
\end{table}


\subsection{Comparison with full architectures in Q-Synth}

\input{figures/qsynth-subarch.tex}

We now evaluate the effectiveness of using subarchitectures in the tool Q-Synth~\cite{Shaik_2024,Q-Synth}.
Q-Synth v2.0 uses SAT solving to find provably optimal layout mappings and can map relatively deep circuits (requiring many SWAPs) to relatively large quantum platforms (100+ qubits). 

In this experiment, we run Q-Synth on the circuits used in~\cite{Shaik_2024}. These include arithmetic circuits (3-16 qubits and 6-90 CNOT gates) from OLSQ~\cite{Tan_2020}, QUEKO benchmarks (16 or 54 qubits, 15-270 CNOT gates) from OLSQ~\cite{Tan_2020}, 
and VQE benchmarks obtained from a company (all 8 qubits and ranging from 18-79 CNOT gates). 
We map these circuits to all 6 quantum platforms mentioned in Table~\ref{tab:subarchitectures}, without ancillaries.
We have also run the experiments with ancillas: allowing 2 ancillas makes no difference; allowing unlimited ancillas saved 1 SWAP in only 4 circuit/platform combinations.

Fig.~\ref{fig:qsynth_plots} displays the time used by Q-Synth to map the circuits on either the full architecture, or on a series of subarchitectures as computed by Alg.~\ref{alg:subarchitectures}.
These timings do not include the time for computing maximal subarchitectures, since these can be precomputed once for each platform (cf.~Table.~\ref{tab:subarchitectures}).
    
From the data (orange dots below diagonal), it is clear that Q-Synth profits from using subarchitectures, in particular for the large platforms (i.e., Sycamore, Rigetti-80 and Eagle, with 53, 80 and 127 qubits). The orange outliers
above the diagonal correspond to quantum circuits with many (16-54) qubits.
On the other hand, for the smaller platforms (16-20 qubits) it is faster to map the circuits directly to the full architecture. In particular for the IBM Tokyo platform, the number of maximal subarchitectures is quite high, leading to many mapping-queries.
A cumulative overview of this data is also available in Table~\ref{tab:cumqsynth}; the underlying data is available in~\cite{KostyantinBSc}. It shows that 
using subarchitectures for large platforms speeds up each benchmark by 10 minutes on average, while for small platforms, it is 3s slower on average.

\begin{table}[tb]
    \caption{Cumulative numbers: Improvement of using maximal subarchitectures in Q-Synth, for small and large platforms.\label{tab:cumqsynth}}
    \centering
    \begin{tabular}{|l||rrr|rr|}
        \hline
        & \multicolumn{3}{|c|}{Absolute improvement} & \multicolumn{2}{|c|}{Relative improvement} \\
    Platforms & min (s) & avg (s) & max (s) &
       geomean & max \\
        \hline   
    Small (N=87)  & -1279.22 & -3.37 & 3009.18  & 0.34x & 4.28x\\
    Large (N=73)  & -3917.43 & 608.57 & 7589.12 & 5.69x & 641.70x \\  
    \hline
    \end{tabular}
    \end{table}

Using subarchitectures provides an interesting new functionality to Q-Synth. 
In the presence of ancillas, Q-Synth synthesizes a quantum layout with the minimal number of SWAP gates, but no bound on the ancilla qubits. With subarchitectures, one can now put a bound on the number of ancillas and minimize the number of SWAP gates within that ancilla bound.

\section{Conclusion and Future Perspectives}
We have reduced the search space of optimal QLS (quantum layout synthesis) of an $n$-qubit quantum circuit on a $p$-qubit physical quantum platform ($n\leq p$).
Our method imposes a bound $k$ on the number of ancilla qubits and computes all maximal subarchitectures of size $n+k$. Alg.~\ref{alg:subarchitectures} improves over
naïve enumeration, by enumerating connected subgraphs directly, and filtering them
with graph hashing and subgraph isomorphism checks. Alg.~\ref{alg:mapping} maps the circuit on all maximal subarchitectures. The search space is further pruned by bounding the number of SWAP gates to the best solution found so far.
This method can be used to find the optimal layout given at most $k$ ancillas. By increasing $k$, the method also guarantees the lowest number of ancillas to achieve the overall optimal layout. In the latter case, one would have to continue until $n+k=p$. Efficiently computing a lower bound on the required ancillas $k$ is an interesting but difficult remaining challenge~\cite{Peham_2023}.

We demonstrated that our method speeds up existing mapping tools QMAP and Q-Synth up to three orders of magnitude,
in particular for larger quantum platforms, while still providing clear optimality guarantees.
It should be stressed that our method is orthogonal to the underlying QLS method; it simply applies an existing QLS tool to each maximal subarchitecture. This paper focused on minimal gate count, but the method can be directly combined with depth-optimal QLS tools, like OLSQ2~\cite{Tan_2021}, preserving guaranteed depth-optimality.

An interesting avenue for future work would be to see if heuristic approaches, like SABRE~\cite{Li_2019} or QMAP with A*~\cite{Wille_2019} would also profit from maximal subarchitectures. Heuristic tools might provide better results when restricted to maximal subarchitectures, but this needs to be investigated.

Finally, we mention two possible extensions of our work: First, one could investigate heuristics to sort the maximal subarchitectures, aiming at finding good solutions earlier.
%
Another improvement would be to take into account the quality of the qubits. 
One could select a maximal subarchitecture of high quality qubits. 
A more difficult question would be to compute a \emph{noise-optimal quantum layout}, where a few extra SWAP gates can be compensated by selecting higher quality qubits.

\subsection*{Acknowledgements}
This paper is based on the BSc thesis of the first author \cite{KostyantinBSc}, supervized by the 2nd and 3rd author.
The research was partially funded by the Innovation Fund Denmark project
``Automated Planning for Quantum Circuit Optimization''.
This research was also partially funded by the European Innovation Council through Accelerator grant no. 190124924.
The experiments were carried out on the Grendel cluster at the
Centre for Scientific Computing, Aarhus (\url{http://www.cscaa.dk/grendel/}).

\bibliographystyle{IEEEtran}
\bibliography{Subarchitectures}

    
\end{document}

%% file: preamble.tex
\usepackage{cite}
\usepackage{amsmath,amssymb,amsfonts}
\usepackage[noend]{algpseudocode}
\usepackage[pdftex]{graphicx} 
\usepackage{textcomp}
\usepackage{xcolor}
\usepackage[a4paper, total={184mm,239mm}]{geometry}
\usepackage{hyperref}
\usepackage{todonotes}
\usepackage[qm]{qcircuit}
\usepackage{caption,subcaption}

\usepackage{algorithm}
\algblockdefx[Foreach]{ForEach}{EndForEach}[1]{\textbf{for each} #1 \textbf{do}}{\textbf{end for}}
\makeatletter
\ifthenelse{\equal{\ALG@noend}{t}}{\algtext*{EndForEach}}{}
\makeatother

\newtheorem{theorem}{\sc{Theorem}}
\newenvironment{proof}{\paragraph*{Proof}}{\hfill$\Box$}

\newcommand{\CX}{\mathrm{CX}}
\newcommand{\SWAP}{\mathrm{SWAP}}

\usepackage{tikz}
\usetikzlibrary{positioning}
\usetikzlibrary{shapes}
\usetikzlibrary{patterns}
\usetikzlibrary{arrows}
\usetikzlibrary{decorations.pathmorphing}
\tikzset{every picture/.style={>=stealth'}}

\tikzstyle{region}     			= [line width=2pt, rounded corners]
\pgfdeclarelayer{background}
\pgfsetlayers{background,main}

\usepackage{graphicx}
\usepackage{pgfplots}

\newcommand{\tikzdot}[1]{
  \protect\tikz[baseline=-3.4pt]{
    \fill[color=#1, opacity=0.7] circle[radius=1.7pt];
  }
}

\colorlet{cblue}{cyan}
\tikzstyle{dots_cblue}=[only marks, mark=*, mark size=1.8pt, mark options={color=cblue, fill=cblue, opacity=0.5}]
\tikzstyle{x_cblue}=[only marks, mark=x, mark size=2.2pt, mark options={color=cblue, fill=cblue, opacity=0.7}]
\tikzstyle{plot_cblue}=[color=cblue, dashed, line width=0.7pt, opacity=0.7]

\colorlet{corange}{orange}
\tikzstyle{dots_corange}=[only marks, mark=*, mark size=1.8pt, mark options={color=corange, fill=corange, opacity=0.5}]
\tikzstyle{x_corange}=[only marks, mark=x, mark size=2.2pt, mark options={color=corange, fill=corange, opacity=0.7}]
\tikzstyle{plot_corange}=[color=corange, dashed, line width=0.7pt, opacity=0.7]

%% file: figures/qmap-cover4.tex
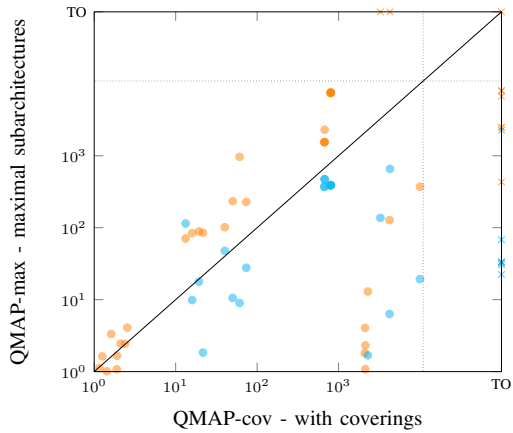
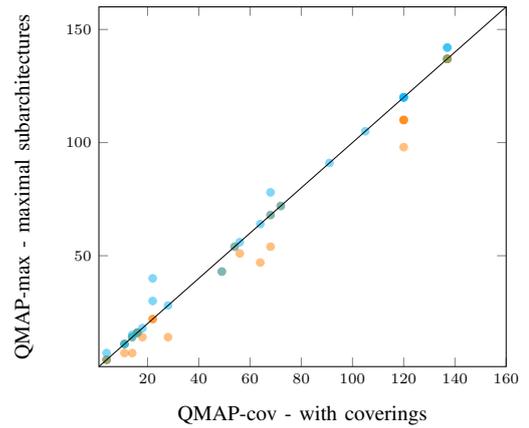
\begin{figure*}[htbp!]
    \centering
      \subfloat[Running Time (s)] {
        \label{fig:mqt_4:time}
        \centering
        \scalebox{0.8}{
        \begin{tikzpicture}
          \begin{axis}[%
            width=0.45\linewidth, height=0.41\linewidth,
            every tick label/.append style={font=\scriptsize},
            xlabel={QMAP-cov - with coverings},
            xmin=1,
            xmax=100000,
            xtick={1,10,100,1000,100000},
            xticklabels = {$10^{0}$, $10^{1}$,$10^{2}$,$10^{3}$,TO},
            xmode=log,
            ylabel={QMAP-max - maximal subarchitectures},
            ymin=1,
            ymax=100000,
            ytick = {1,10,100,1000,100000},
            yticklabels = {$10^{0}$,$10^{1}$,$10^{2}$,$10^{3}$,TO},
            ymode=log,
            grid style={dashed,black!12},
            ]
    
            \addplot[domain=0.001:100000, samples=8, color=black]
            {x};
    
            \draw[densely dotted, opacity=0.4] (9.3,100) -- (9.3,-100.0);
            \draw[densely dotted, opacity=0.4] (100,9.3) -- (-100.0,9.3);
    
            \begin{scope}[blend mode=soft light]
              \addplot+ [forget plot, style=dots_cblue] table {./figures/data/mqt40.tex};
              \addplot+ [forget plot, style=x_cblue] table {./figures/data/mqt40_timeout.tex};
              \addplot+ [forget plot, style=dots_corange] table {./figures/data/mqt41.tex};
              \addplot+ [forget plot, style=x_corange] table {./figures/data/mqt41_timeout.tex};
            \end{scope}
          \end{axis}
        \end{tikzpicture}}
      }
      \qquad\qquad
      \subfloat[Additional Gates] {
        \label{fig:mqt_4:gates}
        \centering
        \scalebox{0.8}{
        \begin{tikzpicture}
          \begin{axis}[%
            width=0.45\linewidth, height=0.41\linewidth,
            every tick label/.append style={font=\scriptsize},
            xlabel={QMAP-cov - with coverings},
            xmin=1,
            xmax=160,
            ylabel={QMAP-max - maximal subarchitectures},
            ymin=1,
            ymax=160,
            grid style={dashed,black!12},
            ]
    
            \addplot[domain=0.001:160, samples=8, color=black]
            {x};
    
            \begin{scope}[blend mode=soft light]
              \addplot+ [forget plot, style=dots_cblue] table {./figures/data/mqt40_gates.tex};
              \addplot+ [forget plot, style=dots_corange] table {./figures/data/mqt41_gates.tex};
            \end{scope}
          \end{axis}
    
        \end{tikzpicture}}
      }
    
    
      \caption{MQT QMAP, mapping to \emph{coverings of size 4} versus \emph{maximal subarchitectures},
      using \tikzdot{cblue} 0 or\tikzdot{corange} 1 ancillary qubits.\vspace{-0.8em}}
      \label{fig:mqt_4}
    \end{figure*}

%% file: figures/qsynth-subarch.tex
\begin{figure}[b]
    \centering
      \label{fig:qsynth}
      \scalebox{0.8}{
      \begin{tikzpicture}
        \begin{axis}[%
            width=0.45\textwidth, height=0.41\textwidth,
            every tick label/.append style={font=\scriptsize},
            xlabel={Full Architecture Mapping Time (s)},
            xmin=0.01,
            xmax=100000,
            xtick={0.1,1,10,100,1000,100000},
            xticklabels = {$10^{-1}$, $10^{0}$, $10^{1}$,$10^{2}$,$10^{3}$,TO},
            xmode=log,
            ylabel={Subarchitectures Mapping Time (s)},
            ymin=0.01,
            ymax=100000,
            ytick = {0.1,1,10,100,1000,100000},
            yticklabels = {$10^{-1}$,$10^{0}$,$10^{1}$,$10^{2}$,$10^{3}$,TO},
            ymode=log,
            grid style={dashed,black!12},
        ]
    
          \addplot[domain=0.001:100000, samples=8, color=black]
          {x};
    
          \draw[densely dotted, opacity=0.4] (9.3,100) -- (9.3,-100.0);
          \draw[densely dotted, opacity=0.4] (100,9.3) -- (-100.0,9.3);
    
          \begin{scope}[blend mode=soft light]
            \addplot+ [forget plot, style=dots_corange] table {./figures/data/qsynth_large.tex};
            \addplot+ [forget plot, style=x_corange] table {./figures/data/qsynth_large_timeout.tex};
    
            \addplot+ [forget plot, style=dots_cblue] table {./figures/data/qsynth_small.tex};
            \addplot+ [forget plot, style=x_cblue] table {./figures/data/qsynth_small_timeout.tex};
    
          \end{scope}
        \end{axis}
    
      \end{tikzpicture}}

    \medskip
    
      \caption{Q-Synth mapping to \emph{full architectures} and to \emph{maximal subarchitectures},
      on {\tikzdot{cblue}} \emph{small} and {\tikzdot{corange}} \emph{large} quantum platforms.}
      \label{fig:qsynth_plots}
    \end{figure}
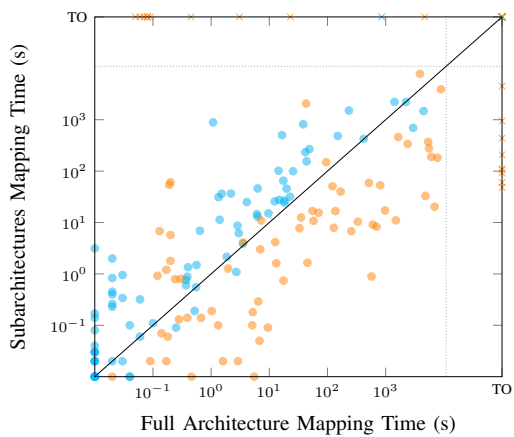